\documentclass[a4paper,10pt]{article}
\usepackage[utf8x]{inputenc}
\pdfoutput=1
\usepackage{nccfoots}

\usepackage{times}
\usepackage{geometry}
\usepackage[T1]{fontenc}
\usepackage{color}
\usepackage{amsmath}
\usepackage{amssymb}
\usepackage{array}
\usepackage{amsthm}
\usepackage{graphicx}
\usepackage{hyperref}
\usepackage{setspace}
\usepackage{stmaryrd}
\usepackage{marginnote}
\usepackage[mathscr]{euscript}

\theoremstyle{plain}
\newtheorem{thm}{Theorem}[section]

\newtheorem{lem}[thm]{Lemma}

\theoremstyle{definition}

\theoremstyle{remark}
\newtheorem{rem}[thm]{Remark}

\theoremstyle{plain}

\newcommand{\R}{\mathbb{R}}

\newcommand{\N}{\mathbb{N}}

\newcommand{\identity}{\mathrm{id}}

\newcommand{\ric}{\mathrm{Ric}}
\newcommand{\trace}{\mathrm{tr}}
\newcommand{\kernel}{\mathrm{ker}}

\newcommand{\Diff}{\mathrm{Diff}}
\newcommand{\image}{\mathrm{im}}

\newcommand{\diver}{\mathrm{div}}

\newcommand{\abg}[1]{|#1|_{g}}

\newcommand{\go}[1]{\/^{\scriptscriptstyle{(#1)}}\!\tilde g}

\newcommand{\eq}[1]{\begin{equation}#1\end{equation}}

\newcommand{\alg}[1]{\begin{aligned}#1\end{aligned}}

\newcommand{\bc}{\begin{cases}}
\newcommand{\ec}{\end{cases}}

\newcommand{\Chr}[3]{\Gamma^{#1}_{#2 #3}}
\newcommand{\Chrh}[3]{\hat{\Gamma}^{#1}_{#2 #3}}

\newcommand{\p}[1]{\partial_{#1}}

\newcommand{\Ab}[1]{\|#1\|}

\newcommand{\ga}{\gamma}

\newcommand{\Si}{\Sigma}
\newcommand{\La}{\Lambda}

\newcommand{\De}{\Delta}

\newcommand{\na}{\nabla}

\newcommand{\mcl}[1]{\mathcal{#1}}
\newcommand{\mcr}[1]{\mathscr{#1}}

\renewcommand{\title}[1]{{\bfseries #1}\par}
\renewcommand{\author}[1]{\medskip{#1}\par\smallskip}

\numberwithin{equation}{section}

\begin{document}

\begin{center}
\title{\Large On the CMC--Einstein--$\Lambda$ flow }
\vspace{3mm}
\author{ David Fajman and Klaus Kröncke}
\date\today
\Footnotetext{}{1991 \emph{Mathematics Subject Classification.} 58J45,53C25,83C05.}
\Footnotetext{}{\emph{Key words and phrases.} Nonlinear Stability, General Relativity, Einstein metrics, Einstein flow.}
\vspace{1mm}
\end{center}

\begin{center}
May 2018
\end{center}
\vspace{2mm}

\begin{abstract}
We complement a recent work on the stability of fixed points of the CMC-Einstein-$\Lambda$ flow. 
In particular, we modify the utilized gauge for the Einstein equations and remove a restriction on the fixed points whose stability we are able to prove by this method, and thereby generalize the stability result. In addition, we consider the notion of the \emph{reduced Hamiltonian}, originally introduced by Fischer and Moncrief for the standard CMC-Einstein flow. For the analog version of the flow in the presence of a positive cosmological constant we identify the stationary points and relate them to the long-time behavior of the flow on manifolds of different Yamabe types. This entails conjectures on the asymptotic behaviour and potential attractors.
\end{abstract}

\section{Introduction}
Determining the long-time behaviour of the Einstein flow on 
compact manifolds without boundary is one of the central objectives of mathematical cosmology. In the presence of a positive cosmological constant\footnote{The positive cosmological constant $\Lambda$ can either be added directly into the Einstein equations via $R_{\mu\nu}-\frac12Rg_{\mu\nu}+\Lambda g_{\mu\nu}=T_{\mu\nu}$, as done in the present paper, or modelled by a scalar field with suitable potential as in \cite{Ri08}. Both models lead to identical long-time behaviour of the spacetime geometry in expanding direction, which is compatible with the observed accelerated expansion.}, this behaviour is well understood for the class of initial data close to homogeneous cosmological models due to the work of Ringstr\"om \cite{Ri08}. This behaviour is independent of the topology of the spatial hypersurfaces of spacetime due to the fast expansion rate in the case $\La>0$, which causes a localization of perturbations in space.

In a recent paper \cite{FK15}, the authors provide an alternative proof of this stability result using the constant-mean-curvature-spatial-harmonic gauge, originally introduced in the work of Andersson and Moncrief \cite{AM03,AM11}. This approach leads to a very concise proof by a suitably arranged energy estimate. However, this method does not cover the full result of Ringstr\"om but only those cases where the perturbed spacetime allows for a CMC foliation with the mean curvature being a time-function. A fundamental difference to the generalized harmonic gauge, used by Ringstr\"om, is the fact that the CMCSH gauge casts the Einstein equations into an elliptic-hyperbolic form. The nature of the elliptic operators determining the equations for lapse and shift depends crucially on the global geometry of spatial hypersurfaces of the foliation. While in the case of negative Einstein geometries these operators are generally invertible, certain restrictions are required in the case of positive curvature due to the operator defining the shift vector field. For the spatial harmonic gauge, defining this operator, it is necessary that the background Riemannian Einstein manifold $(M,\gamma)$ does not admit Killing fields and the Laplacian $\Delta_\gamma$ does not have $-2/nR(\gamma)$ as an eigenvalue. 

It turns out that the assumptions of the stability theorem can be relaxed to include those cases when a modification of the spatial harmonic gauge is employed. This \emph{modified spatial harmonic gauge} is introduced in this paper.  

Another question concerning the CMC-Einstein-$\Lambda$ flow regards the existence of other attractors except for the spatial Einstein geometries. This has been investigated by Fischer and Moncrief using the notion of \emph{reduced Hamiltonian} for the case of vanishing cosmological constants \cite{FM01,FM02}. In the second part of this paper we introduce an analog quantity for the case $\Lambda>0$ and analyze its behaviour along the flow. We then draw some conclusions on the possibility for existence of data not evolving to a spatial Einstein geometry.

\subsection{Modified spatial harmonic gauge}
The existence of an eigenvalue $-2/nR(\gamma)$ in the spectrum of the Laplacian associated with $\gamma$ prevents the elliptic operator in the equation of the shift vector in \eqref{CMC} to be an isomorphism. The form of this operator is a direct consequence of the spatial harmonic gauge condition. The modification of this condition that we employ takes the form 
\eq{
g^{kl}((1+\alpha)(\overline{\nabla}_kg_{il}+\overline{\nabla}_lg_{ik})-\overline{\nabla}_ig_{kl})=0,
}
where $\alpha$ is a small real parameter and $\overline{\nabla}$ is the covariant derivative of $\gamma$. Setting $\alpha=0$ recovers the spatial harmonic gauge condition.

For nontrivial $\alpha$, the elliptic operator appearing in the associated shift equation (which results from taking the time derivative of the gauge condition) depends on $\alpha$.  In particular, the eigenvalue condition to be avoided to assure isomorphy depends on $\alpha$. In case this eigenvalue condition is fulfilled one can change $\alpha$ slightly and obtain the desired non-existence of the problematic eigenvalue. 

To make this work we need however to assure that all the other effects of this change of gauge are still compatible with the well-posedness of the system and the stability analysis. This includes the proof that the decomposition of the Ricci tensor still leads to a well-defined elliptic operator and that the elliptic operator acting on the shift vector field is an isomorphism for suitable values of $\alpha$. Finally, the generalized version of the stability result is proven in Theorem \ref{main-thm}.

We emphasize that the modified spatial harmonic gauge has other potential applications for the CMCSH-Einstein flow as for instance on manifolds in the positive Yamabe class, also in the absence of a positive cosmological constant.

\subsection{Reduced Hamiltonian}
An elegant approach to the question of existence of other attractors for the CMC Einstein flow than negative Einstein geometries was introduced by Fischer and Moncrief in form of the so-called \emph{reduced Hamiltonian}. Considering a CMC-foliation of spacetime, the reduced Hamiltonian $H_R$ of a 3-dimensional spatial hypersurface $M$ is defined as its volume rescaled by the mean curvature $\tau$,
\eq{\label{red-ham-org}
H_R(M)= -\tau^3 \mathrm{vol}_g(M),
}
where $g$ is the induced Riemannian metric on $M$. The remarkable property of $H_R$ is that, for $M$ being in the negative Yamabe class, it is monotonically decreasing in time and constant if and only if $g=\frac32\tau^{-2}\gamma$, where $\gamma$ is a hyperbolic Einstein metric of scalar curvature $-1$ on $M$. For a fixed value of the mean curvature, then $H_R$ has a critical point, if considered as a functional on the reduced phase space (cf.~\cite{FM02}), at $(\gamma,0)$ unique up to isometry, where $(\gamma,0)$ corresponds to the value of the metric and the tracefree-part of the second fundamental form of the initial data set. This is moreover a strict local minimum of $H_R$ modulo isometries. For non-hyperbolic manifolds, critical points of $H_R$ do not exist. Furthermore, the infimum of the reduced Hamiltonian is directly related to the Yamabe constant of $M$.\\
For the negative Yamabe class, the behaviour of the reduced Hamiltonian implies that there are no other attractors than the negative Einstein metrics. This, however, does not imply that all initial data flows towards a negative Einstein metric as singularities may form before.

For the Einstein flow with positive cosmological constant CMC foliations can be used to prove stability of certain cosmological models as recently shown by the authors \cite{FK15}. This approach works in general when the spatial manifold is of negative Yamabe type and under certain restrictions for the reversed CMC Einstein flow when it is of positive Yamabe type. The stability results for $\Lambda>0$, however, provide no information on possible other attractors. 

In the second part of this paper, we introduce a reduced Hamiltonian for solutions to the CMC- and reversed-CMC Einstein flow with positive $\La$. For the CMC Einstein flow on manifolds in the negative Yamabe class we obtain monotonicity of the reduced Hamiltonian and constancy if and only if the metric is negative Einstein. For manifolds in the positive Yamabe class the reduced Hamiltonian is also monotonic. Moreover, it is constant if and only the flow remains in a set of positive constant scalar curvature metrics up to diffeomorphism in the same conformal class and under an additional eigenvalue condition these are homothetic positive Einstein metrics. \\
\section{Preliminaries}
We briefly recall the setup of the CMC-Einstein-$\Lambda$ flow and some explicit model solutions to illustrate the geometric context.
\subsection{The CMCSH-Einstein--$\Lambda$ flow}
We consider a space-time of the form $\R\times M$, where $M$ is a smooth compact $n$-dimensional manifold without boundary. For the Lorentzian metric
we choose the ADM-Ansatz 
\eq{\label{ADM}
\go {n+1}= -N^2dt\otimes dt+g_{ij}(dx^i+X^i dt)\otimes(dx^j+X^j dt),
}
with $\tilde{g}=(N,X,g)$ lapse function, shift vector field and the spatial metric, respectively. The CMCSH gauge with respect to a fixed Riemannian metric $\gamma$ on $M$ is realized by imposing
\eq{\label{CMCSH}\alg{
&\mathrm{tr}_g k\equiv \tau=t,\\
&g^{ij}(\Chr kij-\Chrh kij)\equiv V^k=0,
}}
where $k$ is the second fundamental form and $\Chr kij$, $\Chrh kij$ denote the Christoffel symbols w.r.t~$g$ and $\ga$, respectively. The arbitrary, fixed Riemannian metric $\gamma$ on $M$ serves here as a reference metric to impose the spatial harmonic gauge condition. In the eventual stability analysis (cf.~\cite{FK15}) this metric is then chosen to be the fixed point of the rescaled Einstein-$\Lambda$ flow on $M$.  Setting $\La=\frac{n(n-1)}{2}$ the CMCSH-Einstein-$\Lambda$ flow reads 
\eq{\label{CMC}
\alg{
R(g)-\abg{\Si}^2+\tau^2\left(\frac{n-1}{n}\right)&=n(n-1)\\
\na^i\Si_{ij}&=0\\
\p tg_{ij}&=-2N(\Si_{ij}+\tau/ng_{ij})+\mcr{L}_{X}g_{ij}\\
\p t\Si_{ij}&=N(R_{ij}+\tau \Si_{ij}-2\Si_{ik}\Si_j^k+(\tau^2/n-n)g_{ij})\\
&\quad+\mcl L_X\Si_{ij}-\frac1ng_{ij}-\frac{2N\tau}{n}\Si_{ij}-\na_i\na_jN\\
\De N&=-1 + N\Big[\abg{\Si}^2+\frac{\tau^2}{n}-n\Big]\\
\Delta X^i+R^i_{\,m}X^m-\mcl L_XV^i&=2\nabla_jN\Si^{ji}+\tau(2/n-1)\nabla^iN\\
&\quad-(2N\Si^{mn}- (\mcl L_Xg)^{mn})(\Chr imn-\hat{\Gamma}^i_{mn})
}
}
where the second fundamental form $k$ is decomposed by $k=\Si+\frac{\tau}n g$,
where $\Si$ denotes the trace-free part. Furthermore, $R(g)$ is the Ricci scalar curvature of $g$, $\mathcal L_X$ denotes the Lie derivative w.r.t.~the shift and $R_{ij}$ denotes the Ricci tensor of the metric $g$. The Laplacian $\Delta$ is defined w.r.t.~$g$.

The system \eqref{CMC} is derived as follows: The first two equations are the Hamiltonian and the momentum constraint. The third equation is a consequence of the ansatz \eqref{ADM}. The fourth equation follows from demanding Einstein's equation to hold. The fifth equation follows from taking the trace of the fourth equation and using the first equation in \eqref{CMCSH} and the Hamiltonian constraint. The sixth equation follows from taking the time derivative of the second equation in \eqref{CMCSH}.
For details, see e.\ g.\ \cite{Re08}

In the case of an reversed CMC-gauge, $t=-\tau$, which we use for spatial Einstein metrics of positive curvature, the lapse equation takes the form
\eq{\label{inverseCMC}
\De N=1 + N\Big[\abg{\Si}^2+\frac{\tau^2}{n}-n\Big].
}
The equation for the trace-free part of the second fundamental form then reads
\eq{\alg{
\p t\Si_{ij}&=N(R_{ij}+\tau \Si_{ij}-2\Si_{il}\Si_j^l+(\tau^2/n-n)g_{ij})\\
&\quad+\mcl L_X\Si_{ij}+\frac1ng_{ij}-\frac{2N\tau}{n}\Si_{ij}-\na_i\na_jN.
}}

\subsection{Model solutions}\label{ssec : 24}
Some model solutions, which are fixed points of the flow and which have been investigated regarding their future stability in \cite{FK15}, are recalled in the following. By fixed points in this context, we mean that $\Sigma=0$ and $X=0$.

In case of the CMC-gauge, $t=\tau$, the lapse equation, for $\Si=0$, is solved by
\eq{
N=\frac{n}{\tau^2-n^2}
}
and since $N>0$, $\tau^2>n^2$. Fix $\tau_0\in (-\infty,-n)\cup (n,\infty)$.
Then the physical metric is given by
\eq{\label{phys-metr}
g(\tau)=g(\tau_0)\frac{\tau_0^2-n^2}{\tau^2-n^2},
}
where the metrics have the property that the Ricci tensor is given by $R_{ij}=-\frac{n-1}{nN}g_{ij}$. 

In reversed CMC-gauge, $t=-\tau$. For $\Si=0$, then

\eq{
N=\frac{n}{n^2-\tau^2}
}
and therefore, $\tau^2<n^2$. Fix $\tau_0\in (-n,n)$. Then the solution for the physical metric is again \eqref{phys-metr}, and
$R_{ij}=\frac{n-1}{nN}g_{ij}$.

Observe that in the case $\tau=n$, the equation on the Lapse function can not be solved.



\section{The modified Harmonic gauge}
In this section we introduce the modified \emph{constant-mean-curvature-spatial-harmonic gauge} and use it to generalize Theorem 1.2 from \cite{FK15}. The main result is provided in Theorem \ref{main-thm}.
\subsection{The modified spatial harmonic gauge}
Recall that the spatial harmonic gauge condition with respect to a fixed Riemannian metric $\gamma$ is given by
\begin{align}g^{ij}(\Gamma_{ij}^k-\overline{\Gamma}_{ij}^k)=0,
\end{align}
where $\Gamma^{ij}_k,\overline{\Gamma}_{ij}^k$ are the  Christoffel symbols of $g,\gamma$, respectively.
This can also be reformulated as
\begin{align}g^{kl}(\overline{\nabla}_kg_{il}+\overline{\nabla}_lg_{ik}-\overline{\nabla}_ig_{kl})=0.
\end{align}
Here and throughout this section, $\overline{\nabla}$ is the covariant derivative of $\gamma$.
We will now introduce $\alpha\in\R$ as a parameter and impose the gauge condition
\begin{align}
g^{kl}((1+\alpha)(\overline{\nabla}_kg_{il}+\overline{\nabla}_lg_{ik})-\overline{\nabla}_ig_{kl})=0.
\end{align}
This ansatz is technically motivated and allows for a direct control of the corresponding kernel of the elliptic operator on the left-hand side of \eqref{shift-eq} acting on the shift vector field via the parameter $\alpha$.
Let $\mathcal{M}$ be the set of smooth Riemannian metrics on $M$.
Fix a metric $\gamma\in \mathcal{M}$ and let $\mathcal{H}_\alpha$ be the set of metrics satisfying this modified gauge condition, i.e.\
\begin{align}\label{SHmetrics}
\mathcal{H}_{\alpha}=\left\{g\in\mathcal{M}\mid (V_{\alpha,g})_i:=g^{kl}((1+\alpha)(\overline{\nabla}_kg_{il}+\overline{\nabla}_lg_{ik})-\overline{\nabla}_ig_{kl})=0\right\},
\end{align}
\subsection{Expansion of the Ricci tensor in the modified Harmonic gauge}
If $g\in \mathcal{H}_{\alpha}$, we have
\begin{align}\begin{split}
0&=\overline{\nabla}_i(V_{\alpha,g})_j+\overline{\nabla}_j(V_{\alpha,g})_i\\
&=g^{kl}[2(1+\alpha)(\overline{\nabla}^2_{jk}g_{il}+\overline{\nabla}^2_{ik}g_{jl})-(\overline{\nabla}^2_{ji}g_{kl}+\overline{\nabla}^2_{ij}g_{kl})]+g^{-1}*g^{-1}*\overline{\nabla}g*\overline{\nabla}g\\
\end{split}
\end{align}
where $h=g-\gamma$. By \cite[p.\ 234, (53)]{Sh89}, we have
\begin{align}\begin{split}
-2R_{ij}&=g^{kl}(\overline{\nabla}^2_{kl}g_{ij}+\overline{\nabla}^2_{ij}g_{kl}
-\overline{\nabla}^2_{il}g_{jk}-\overline{\nabla}^2_{kj}g_{il})\\
&\quad
-g^{kl}g_{jp}\tilde{g}^{pq}\overline{R}_{ikql}-\overline{R}_{ij}+g^{-1}*g^{-1}*\overline{\nabla}g*\overline{\nabla}g\\
&=g^{kl}(\overline{\nabla}^2_{kl}h_{ij}+\overline{\nabla}^2_{ij}h_{kl}
-\overline{\nabla}^2_{il}h_{jk}-\overline{\nabla}^2_{kj}h_{il})\\
\end{split}
\end{align}
where $\overline{R}_{ikql},\overline{R}_{ij}$ denote the Riemann curvature and the Ricci curvature of $\gamma$, respectively.
Adding up these two equation shows that
\begin{align}\begin{split}
-2R_{ij}&=g^{kl}(\overline{\nabla}^2_{kl}g_{ij}+\alpha(\overline{\nabla}^2_{jk}g_{il}+\overline{\nabla}^2_{ik}g_{jl})+\overline{\nabla}^2_{ij}g_{kl}-\overline{\nabla}^2_{ji}g_{kl})\\
&\quad-g^{kl}g_{jp}\tilde{g}^{pq}\overline{R}_{ikql}-\overline{R}_{ij}+g^{-1}*g^{-1}*\overline{\nabla}g*\overline{\nabla}g.
\end{split}
\end{align}

Similarly as in \cite[p.\ 234, (55)]{Sh89}, one obtains after commuting covariant derivatives
\begin{align}\begin{split}
-2R_{ij}&=g^{kl}\overline{\nabla}^2_{kl}g_{ij}+\alpha g^{kl}(\overline{\nabla}^2_{jk}g_{il}+\overline{\nabla}^2_{ik}g_{jl})\\
&\quad
-g^{kl}g_{ip}\gamma^{pq}\overline{R}_{jkql}-g^{kl}g_{jp}\gamma^{pq}\overline{R}_{ikql}
+g^{-1}*g^{-1}*\overline{\nabla}g*\overline{\nabla}g.
\end{split}
\end{align}
The consequence is the following 
\begin{lem}\label{Ricciexpansion}
	Let $(M,\gamma)$ be an Einstein manifold and $g\in \mathcal{H}_\alpha$. Then its Ricci tensor admits the expansion
	\begin{align}
	R_{ij}=\frac{1}{2}\mathcal{L}_{g,\gamma,\alpha}(g-\gamma)+J_{ij}
	\end{align}
	where $\mathcal{L}_{g,\gamma,\alpha}$ is a second order differential operator acting on symmetric $2$-tensors, defined by
	\begin{align}
	\mathcal{L}_{g,\gamma,\alpha}h_{ij}=-\frac{1}{\mu_g}\overline{\nabla}_k(g^{kl}\mu_g\overline{\nabla}_lh_{ij})-\frac{\alpha}{\mu_g}[\overline{\nabla}_j(g^{kl}\mu_g\overline{\nabla}_kh_{il})+\overline{\nabla}_i(g^{kl}\mu_g\overline{\nabla}_kh_{jl})]-2\overline{R}_{ikjl}h^{kl}
	\end{align}
	and $J$ is a symmetric tensor depending on $g,\gamma$ which satisfies the estimate
	\begin{align}
	\left\|J\right\|_{H^s}\leq C\left\|g-\gamma\right\|_{H^{s+1}}^2
	\end{align}
	for every $s\in \N_0$. Here, $H^s$ denotes the Sobolev norm of order $s$.
	\end{lem}

\begin{lem}\label{lem-ell-op}
	Let $(M,\gamma)$ be an Einstein manifold and suppose that $\alpha>-\frac{1}{2}$.
Then the operator $\mathcal{L}_{g,\gamma,\alpha}$ is an elliptic operator of second order which is self-adjoint with respect to the scalar product
\begin{align}
(h,\bar{h})_{L^2(g,\gamma)}=\int_M \gamma^{ik}\gamma^{jl}h_{ij}\bar{h}_{kl} \cdot \mu_g.
\end{align}
	\end{lem}
\begin{proof}
	Self-adjointness is clear. To prove ellipticity, we consider the principal symbol of the operator which is
	\begin{align}
	\sigma(\mathcal{L}_{g,\gamma,\alpha},\xi)(h)_{ij}=-g^{kl}\xi_k\xi_lh_{ij}-\alpha g^{kl}(\xi_j\xi_kh_{il}+\xi_i\xi_kh_{jl}).
	\end{align}
	We have
	\begin{align}
	\langle 	\sigma(\mathcal{L}_{g,\gamma,\alpha},\xi)(h),h\rangle_g=-|\xi|_g^2|h|_g^2+2\alpha | h(\xi,.)|^2\leq -(1-2\alpha)|\xi|_g^2|h|_g^2
	\end{align}
	so that	$\sigma(\mathcal{L}_{g,\gamma,\alpha},\xi)$ is an isomorphism for $\xi\neq 0$ if $\alpha>-\frac{1}{2}$. This proves the lemma.
\end{proof}

\subsection{A slice theorem for the modified Harmonic gauge}
Our aim in this section is to prove that under certain conditions on the background metric $\gamma$, $\mathcal{H}_{\alpha}$ is a smooth submanifold of $\mathcal{M}$ and a local slice of the action of the diffeomorphism group through $\gamma$. Let $\alpha\in\R$ be fixed.
By the first variation of the Christoffel symbols (see e.g.\ \cite[Theorem 1.174]{Bes08}), the differential of the map $\Phi:g\mapsto V_{\alpha,g}$ at $\gamma$ is given by
\begin{align}
d\Phi_{\gamma}(h)_i=2(1+\alpha){\gamma}^{kl}\overline{\nabla}_kh_{li}-\overline{\nabla}_i\trace_{\gamma} h.
\end{align}
where $h$ is a symmetric 2-tensor.
\begin{lem}\label{Piso}Let $(M,\gamma)$ be an Einstein manifold such that
	$-2/n\cdot R(\ga)\frac{1+\alpha}{1+2\alpha}$ is not an eigenvalue of the Laplacian $\Delta_\ga$ on functions and $\gamma$ does not admit Killing vector fields. Then the operator
	
	\begin{align}
	P:\omega_i\mapsto \overline{\Delta} \omega_i+R_j^{i}[\ga]\omega^j+\frac{\alpha}{1+\alpha}\overline{\nabla}_i{\mathrm{div}_\gamma}\omega
	\end{align}
	is an isomorphism which preserves the decomposition
	\begin{align}\Omega^1(M)=\left\{d f\mid f\in C^{\infty}(M)\right\}\oplus \left\{ \omega\in \Omega^1(M)\mid {\diver}_{\gamma} \omega=0\right\}.
	\end{align}
\end{lem}
\begin{proof} 
	By a standard argument using commutators of covariant derivatives, we have
	\begin{align}\begin{split}
	P\overline{\nabla}_if=\overline{\Delta} \overline{\nabla}^if+R_j^{i}\overline{\nabla}^jf+\frac{\alpha}{1+\alpha}{\overline{\nabla}}_i\overline{\Delta} f&=\frac{1+2\alpha}{1+\alpha}\overline{\nabla}^i\overline{\Delta} f+2R_j^{i}\overline{\nabla}^if\\
	&=\frac{1+2\alpha}{1+\alpha}\overline{\nabla}^i\overline{\Delta} f+\frac{2 R}{n}\cdot \overline{\nabla}^if
    \end{split}
	\end{align}
	which shows that because of the eigenvalue assumption, $P$ maps the first factor bijectively onto itself. By self-adjointness of $P$, the second factor is also preserved. We define maps $L$ and $L^*$ by 
	\begin{align}
	L: \omega\mapsto\frac{1}{2}(\overline{\nabla}_i\omega_j+\overline{\nabla}_j\omega_i),\qquad L^*:h\mapsto-\overline{\nabla}^jh_{ji}.
	\end{align}
	Note that $L^*$ is the adjoint map of $L$ with respect to the $L^2$-scalar product induced by $\ga$. Now for any one-form $\omega$ with $\diver \omega=0$, we have
	\begin{align}-(L^*L \omega)_i=
	\overline{\nabla}^j\overline{\nabla}_j\omega_i+\overline{\nabla}^j\overline{\nabla}_i\omega_j
	=\overline{\Delta} \omega_i+\overline{\nabla}^j\overline{\nabla}_i\omega_j-\overline{\nabla}_i\overline{\nabla}^j\omega_j
	=\overline{\Delta} \omega_i+R_{ij}\omega^j=(P\omega)_i.
	\end{align}
	Thus, $P\omega=0$ implies $L\omega=0$. But the kernel of $L$ is mapped isomorphically to the space of Killing vector fields by the musical isomorphism, and hence, $\omega=0$. Therefore, $P$ is injective and by self-adjointness, $P$ is also surjective. 
\end{proof}
\begin{lem}\label{splittinglemma}
	Let $(M,\gamma)$ be an Einstein manifold such that
	$-2/n\cdot R(\ga)\frac{1+\alpha}{1+2\alpha}$ is not an eigenvalue of the Laplacian $\overline{\Delta}$ and $\gamma$ does not admit Killing vector fields.
	Then, $d\Phi_{\gamma}:C^{\infty}(S^2M)\to C^{\infty}(T^*M)$ is surjective.
	Moreover, we have the splitting
	\begin{align}\label{splitting}
	C^{\infty}(S^2M)=\kernel( d\Phi_{\gamma})\oplus \image(L).
	\end{align}
	Here, $S^2M$ denotes the bundle of symmetric $2$-tensors.
\end{lem}
\begin{proof}
	A computation using commuting covariant derivatives proves 
	\begin{align}\label{bochner}
	d\Phi_{\gamma}\circ L(\omega)_i=(1+\alpha)P\omega_i,
	\end{align}
    so that the first assertion follows from Lemma \ref{Piso}.
    Let now $h\in C^{\infty}(S^2M)$. Again by Lemma \ref{Piso} there exists a unique solution $\omega$ of the equation 
    \begin{align}
    d\Phi_{\gamma}(h)=(1+\alpha)P\omega
    \end{align}
    so that $h-L\omega\in \ker d\Phi_{\gamma}$ by \eqref{bochner} which proves the second assertion.
\end{proof}
\noindent For our purposes, it is more convenient to work in neighbourhoods with Sobolev regularity. We therefore use $H^s$-norms with $s>\frac{n}{2}+1$ for the following theorem. We remark that the above lemmas also hold, if we descend to $H^s$-regularity. Let $\mathcal{M}^s$ be the space of $H^s$-metrics on $M$ and let $\mathcal{H}^s_{\alpha}$ be the set of all $g\in \mathcal{M}^s$ satisfying the condition in \eqref{SHmetrics}. 
\begin{thm}\label{slicethm}
	Let $(M,\gamma)$ be an Einstein manifold such that
	$-2/n\cdot R(\ga)\frac{1+\alpha}{1+2\alpha}$ is not an eigenvalue of the Laplacian $\overline{\Delta}_\ga$ and $\gamma$ does not admit Killing vector fields. Then in a small $H^s$-neighbourhood $\mathcal{U}\subset \mathcal{M}^s$ of $\gamma$, $\mathcal{H}^s$ is a smooth submanifold of $\mathcal{M}^s$ with tangent space
	\begin{align}
	T_{\gamma}\mathcal{H}^s_{\alpha}=\left\{h\in H^s(S^2M)\mid
	d\Phi_{\gamma}(h)_i=2(1+\alpha)\overline{\nabla}^j h_{ji}-\overline{\nabla}_i\trace_{\gamma} h =0\right\}.
	\end{align}
	Moreover, for any $g\in\mathcal{U}$ there exists an isometric metric $\tilde{g}\in\mathcal{H}^s_{\alpha}$ which is $H^s$-close to $\gamma$, i.e.\ there exists $\varphi\in H^s(\Diff(M))$ such that $g=\varphi^*\tilde{g}$.
\end{thm}
\begin{proof}
	The first assertion follows from the first assertion of Lemma \ref{splittinglemma}.
	The second assertion follows from the implicit function theorem for Banach manifolds applied to the map
	\begin{align}\Psi:\mathcal{H}^s_{\alpha}\times  H^s(\Diff(M))\to \mathcal{M}^s
	\end{align}
	given by $\Psi(g,\varphi)=\varphi^*g$.
	Since there are no Killing fields, $d\Psi_{(\gamma,\identity)}$ is injective and its image is
	\begin{align}
	\image(d\Psi_{(\gamma,\identity)})=T_{\gamma}\mathcal{H}^s_{\alpha}\oplus\left\{\mathcal{L}_X\gamma\mid X\in H^s(TM)\right\}=\kernel( d\Phi_{\gamma})\oplus\image(L)
	\end{align} 
	which equals $H^s(S^2M)=T_{\gamma}\mathcal{M}^s$ by \eqref{splitting}.
	Therefore, $\Psi$ is a diffeomorphism from a $H^s$-neighbourhood of $(\gamma,\identity)$ in $\mathcal{H}^s_{\alpha}\times  H^s(\Diff(M))$ to a $H^s$-neighbourhood of $\gamma$ in $\mathcal{M}^s$.
\end{proof}
\begin{rem}
	The assertions of Theorem \ref{slicethm} hold for any Riemannian metric $\gamma$ where the operator $P$ is an isomorphism.
\end{rem}
\subsection{The CMC-Einstein flow with modified Harmonic gauge}
In the CMCMSH (constant mean curvature modified spatial harmonic) gauge
\eq{\alg{
		&\mathrm{tr}_g k\equiv \tau=-t,\quad g\in\mathcal{H}_{\alpha},
}}
where $\Chr kij$, $\Chrh kij$ denote the Christoffel symbols w.r.t~$g$ and $\ga$, respectively, with positive cosmological constant $\La=\frac{n(n-1)}{2}$ reads as \eqref{CMC} with the shift equation replaced by
\eq{\label{shift-eq}
	\alg{
				\Delta X^i+R^i_{\,m}X^m+\frac{\alpha}{1+\alpha}\nabla^i\nabla_jX^j&=2\nabla_jN\Si^{ji}+\tau(\frac{2}{n}-\frac{1}{1+\alpha})\nabla^iN\\
		&\quad
		-(2N\Si- (\mcl L_Xg)*(\Gamma -\hat{\Gamma}).
	}
}
The only equation that differs from the standard CMCSH Einstein flow case is this one which is obtained by differentiating the gauge condition in time and using the evolution equation on $g$.
Note that the left hand side of the equation on $X$ defines an operator which is a perturbation of the operator $P$ in Lemma \ref{Piso}. Thus, it is also an isomorphism under the conditions of this lemma, provided that $g$ is close enough to $\gamma$.
Recall also that in this gauge, the Ricci tensor can be expanded as in Lemma \ref{Ricciexpansion}.
Using the CMCMSH Einstein flow, one is now able to prove the following theorem:
\begin{thm}\label{main-thm}
	Let $M$ be a smooth compact n-dimensional manifold ($n\geq 2$) without boundary and $\gamma$ be an Einstein metric satisfying ~$\ric(\ga)=(n-1)\ga$ which does not admit Killing vector fields . Then for $s>n/2+2$, $s'>n/2+s$ and $\varepsilon>0$ there exists a $\delta(\varepsilon)>0$ s.t.~for initial data $(g_0,k_0)$ satisfying
	\eq{
		\left\| g_0-\gamma\right\|_{H^{s'}}+\left\|k_0\right\|_{H^{s'-1}}<\delta
	}  
	its maximal globally hyperbolic development under the Einstein flow with positive cosmological constant $\Lambda=\frac{n(n-1)}{2}$  can be globally foliated by CMC-hypersurfaces $M_t$, $t\in\R$ such that the induced metrics $g_t$ satisfy
	\eq{
		\left\|\cosh^{-2}(t) g_t-\gamma\right\|_{H^{s}}<\varepsilon.
	}  
In particular, all corresponding homogeneous solutions are orbitally stable and the future-  and past developments of small perturbations are future- and past geodesically complete, respectively.
\end{thm}
Observe that the theorem is almost the same as \cite[Theorem 1.2]{FK15} but we got rid of the condition that $-2(n-1)$ is not an eigenvalue of the Laplacian.
\begin{proof}[Sketch of proof]
	Pick an $\alpha>-\frac{1}{2}$ such that $-2/n\cdot R(\ga)\frac{1+\alpha}{1+2\alpha}$ is not an eigenvalue of the Laplacian $\Delta_\ga$ on functions. Then we have well-posedness for the CMCMSH-Einstein flow with respect to this parameter $\alpha$. From now on, the procedure is as in [FaKr15] we rescale the solution of the Einstein flow and call the rescaled solution $(g,\Sigma,N,X)$. Then for $s>n/2+1$ we define the main total energy.
\begin{align}\begin{split}
				\mathbf E_s(g,\Sigma)&\equiv \Ab{g-\gamma}_{L^2(\gamma)}^2+\sum_{k=0}^{s-1}(\Sigma,(-\Delta_{g,\gamma,\alpha})^k\Sigma)_{L^2(g,\ga)}\\&\quad+\frac14\sum_{k=1}^s(g-\gamma,(-\Delta_{g,\gamma,\alpha})^k(g-\gamma))_{L^2(g,\ga)}
				\end{split}
				\end{align}
Here,
\begin{align}
	\Delta_{g,\gamma,\alpha}h_{ij}=\frac{1}{\mu_g}\overline{\nabla}_k(g^{kl}\mu_g\overline{\nabla}_lh_{ij})+\frac{\alpha}{\mu_g}[\overline{\nabla}_j(g^{kl}\mu_g\overline{\nabla}_kh_{il})+\overline{\nabla}_i(g^{kl}\mu_g\overline{\nabla}_kh_{jl})],
\end{align}
which is an elliptic operator due to Lemma \ref{lem-ell-op} and the choice of $\alpha>-\tfrac12$.	
Note that the energy is equivalent to the $H^s\times H^{s-1}$-norm of $(g,\Sigma)$ since $	-\Delta_{g,\gamma,\alpha}$ is a nonnegative operator. Now if $s>n/2+1$ and $(g,\Sigma)\in H^s\times H^{s-1}$ be a solution of the rescaled CMCMSH Einstein flow, there exists an $\varepsilon>0$ such that for
\eq{
	(g,\Si)\in\mcl B^s_{\varepsilon}(\ga)\times\mcl B^{s-1}_{\varepsilon}(0),
}
the estimate 
\eq{\label{en-est}
	\partial_{\mathrm T} \mathbf E_s(g,\Sigma)\leq \frac{C(\varepsilon)}{\cosh (\mathrm T)}\mathbf E_s(g,\Sigma)
}
holds. The other parts of the proof are exactly as in \cite{FK15}.
\end{proof}

\section{Reduced Hamiltonian}\label{sec : redham}
In the following section we introduce a \emph{reduced Hamiltonian} for the CMC-Einstein-$\Lambda$ flow and prove its monotonicity along the flow. Results of this type are interesting due to the fact that they hold for arbitrary, not necessarily small data. They might therefore be relevant in a study of large initial data, which is in a reasonable sense far from the background solutions. 

\subsection{Monotonicity and stationary points}
\begin{thm}
Let $g_t$ be a solution of the CMC Einstein flow with positive cosmological constant $\Lambda=\frac{n(n-1)}{2}$ which is expanding (i.e. $\tau<0$). Then, if $\tau\in (-\infty,-n)$ , the reduced Hamiltonian\footnote{The reduced Hamiltonian in \eqref{red-ham}, relevant in the present setting $\Lambda>0$, corresponds to the reduced Hamiltonian \eqref{red-ham-org} for the case $\Lambda=0$.}
\eq{\label{red-ham}
H^{-}_{red}=\left(\frac{\tau^2}{n}-n\right)^{n/2}\mathrm{vol}_{g_t} 
}
is monotonically decreasing and stays constant if and only if $g_t$ is a family of (up to isometry) homothetic metrics which are Einstein metrics of negative scalar curvature.
\end{thm}
\begin{proof}
The time-derivative of the reduced Hamiltonian is
\eq{
\partial_{t}H^{-}_{red}=\partial_{\tau}H^{-}_{red}=
\left(\frac{\tau^2}{n}-n\right)^{n/2-1}\tau
\left(\mathrm{vol}_{g_t}-(\frac{\tau^2}{n}-n)\int_M N \mu_{g_t}\right).
}
By integrating the equation on the lapse function, we can replace the right hand side and we get
\eq{
\partial_{t}H^{-}_{red}=\left(\frac{\tau^2}{n}-n\right)^{n/2-1}\tau\int_M N|\Sigma|^2\mu_{g_t}\leq0
}
which proves the first part of the theorem.
Suppose now that the Hamiltonian stays constant.
Then $\Sigma=0$, since $N$ is positive. Since $\Sigma=0$ holds on
any time interval, where the reduced Hamiltonian is constant, we also have $\partial_t\Si=0$ on this interval. Since the operator $\Delta-(\frac{\tau^2}{n}-n)$
is invertible, the equation on the lapse function is uniquely solvable. Thus, $N=(\frac{\tau^2}{n}-n)^{-1}$.
If we now insert $\Sigma=\partial_t\Sigma=0$ in the evolution equation for $\Sigma$, we see that the equation $R_{ij}=-N^{-1}\frac{n-1}{n}g_{ij}$ holds for all time.
It follows from the evolution equation of the metric, that the $g_t$ are homothetic up to isometry.
\end{proof}

\begin{thm}\label{theorem2}
Let $g_t$ be a solution of the reversed CMC Einstein flow with positive cosmological constant $\Lambda=\frac{n(n-1)}{2}$ which is expanding (i.e. $\tau<0$). Then, if $\tau\in (-n,n)$ the reduced Hamiltonian
\eq{
H^{+}_{red}=\left(n-\frac{\tau^2}{n}\right)^{n/2}\mathrm{vol}_{g_t}}
is monotonically increasing. If $H^{+}_{red}$ is constant along the flow, $g_t$ is a family of constant positive scalar curvature metrics which lie up to diffeomorphism in the same conformal class.
The Ricci tensor of $g_t$ satisfies the equation
\begin{align}\label{eq_Ricci}
N(R_{ij}+(\frac{\tau^2}{n}-n)g_{ij})+\frac{1}{n}g_{ij}-\nabla^2_{ij}N=0.
\end{align}
 Moreover, if $-R(g_t)/(n-1)$ is not an eigenvalue of $\Delta_{g_t}$, 
$g_t$ is a family of homothetic metrics which are Einstein metrics of positive scalar curvature.
\end{thm}
\begin{proof}
The time derivative is
\eq{
\partial_{t}H^{+}_{red}=\partial_{-\tau}H^{+}_{red}=
\left(n-\frac{\tau^2}{n}\right)^{n/2-1}(-\tau)
\left(-\mathrm{vol}_{g_t}+(n-\frac{\tau^2}{n})\int N \mu_{g_t}\right).
}
By integrating \eqref{inverseCMC}, we can replace the right hand side by
\eq{
\partial_{t}H^{+}_{red}=\left(n-\frac{\tau^2}{n}\right)^{n/2-1}(-\tau)\int_M N|\Sigma|^2\mu_{g_t}\geq0,
}
which proves the monotonicity.
As above, $\Sigma=\partial_t\Sigma_t=0$ as long as the Hamiltonian stays constant. By the Hamiltonian constraint and the evolution equation of the metric,
\eq{R(g)=(n-1)(n-\frac{\tau^2}{n}),\qquad \partial_tg=-2N\frac{\tau}{n}g+\mathcal{L}_X g,}
which proves the first assertion in the equality case since the equation on the Ricci tensor follows immediately from the evolution equation of $\Sigma$.

 If $-R(g_t)/(n-1)$ is not an eigenvalue of $\Delta_{g_t}$, the equation on the lapse is uniquely solvable and so, 
$N=(n-\frac{\tau^2}{n})^{-1}$. By the evolution equation for $\Sigma$, this immediately yields
$R_{ij}=N^{-1}\frac{n-1}{n}g_{ij}$. As in the previous theorem, the metrics $g_t$ are homothetic up to isometry in this case.
\end{proof}
\begin{rem}
	It would be interesting to see if there are nontrivial examples (i.e.\ not Einstein) of constant scalar cuvature metrics on a compact manifold which satisfy \eqref{eq_Ricci}.
\end{rem}

\subsection{Critical points on the reduced phase space}
Since these Hamiltonians are (up to the time-dependant scale factor) the same as in the case of vanishing cosmological constant \cite{FM02}, one can also do the same analysis as in this paper.
Suppose the  manifold $M$ is of negative Yamabe-type, i.e.\ the scalar curvature of any metric on $M$ is negative somewhere. Then in any conformal class, there is a unique metric of constant scalar curvature $-1$. By the standard conformal method, we may therefore consider the reduced phase space
\eq{\mathcal{P}^{-}_{red}=\left\{(g,\Sigma)\in\mathcal{M}\times \Gamma(S^2M)\mid R(g)=-1,\trace_g\Sigma=0, (\nabla_g)^i\Sigma_{ij}=0 \right\}.
}
If we drop the assumption of negative Yamabe-type, then the uniqueness statement about metrics of fixed scalar curvature in a conformal class fails. On the other hand, given a metric $g$ close enough to an Einstein metric of constant scalar curvature $1$ which is not the round sphere,
then there exists a unique metric of scalar curvature $1$ in the conformal class of $g$ \cite[Theorem C]{BWZ04}. Therefore, we may in this case introduce the phase space 
\eq{\mathcal{P}^{+}_{red}=\left\{(g,\Sigma)\in\mathcal{M}\times \Gamma(S^2M)\mid R(g)=1,\trace_g\Sigma=0, (\nabla_g)^i\Sigma_{ij}=0 \right\}.
}
One may now regard the reduced Hamiltonians $H^{\pm}_{red}$ as functionals on the Phase spaces $\mathcal{P}^{\pm}_{red}$ and develop their variational theory as in \cite[Section 4]{FM02}. The critical points are the pairs $(g,\Sigma)$ where $\Sigma=0$ and $\ric_{g}=\pm \frac{1}{n}g$, respectively. The Hessian of $H^{-}_{red}$ at a critical point will be positive semidefinite if and only if the Einstein operator $\mathcal{L}_{g,g}$ is positive semidefinite on $TT$-tensors (transverse and traceless tensors) whereas the Hessian of $H^{+}_{red}$ at critical points will be negative semidefinite if and only if $\mathcal{L}_{g,g}$ is positive semidefinite on $TT$-tensors.

 In the case of negative Yamabe-type, the global infimum of $H^{-}_{red}$, will furthermore determine the topological sigma constant which is also known as the Yamabe invariant  (c.f.\ \cite[Theorem 6]{FM02}). In the case of positive Yamabe-type, this is however not true because the phase space is only defined locally around a small neighbourhood of a positive Einstein metric.
 
 By the moniticity property of the
  reduced Hamiltonian, this suggest that Einstein spaces are attractors of the Einstein flow with positive cosmological constant if the Einstein operator  is positive semidefinite on $TT$-tensors.

\subsection{Remarks}
The features of the reduced Hamiltonian for the CMC-Einstein-$\Lambda$ flow and the reversed version, respectively, imply heuristic conclusions for the long time behaviour of these flows for arbitrary initial data. For the negative Yamabe case, the only stable fixed points seem to be negative Riemannian Einstein metrics. For large data, however, it may be that the flow is incomplete in both time directions, i.e. singularities form. 

The similar statement holds for the positive Yamabe class, where Einstein manifolds are replaced by the more general conditions listed in Theorem \ref{theorem2}.

\subsection*{Acknowledgements}
D.F. acknowledges support of the Austrian Science Fund (FWF) project P29900-N27 \emph{Geometric Transport equations and the non-vacuum Einstein-flow}. We thank the anonymous referees for suggestions.

\vspace{0.1cm}
\noindent
\textsc{David Fajman\\
Faculty of Physics, University of Vienna,\\
Boltzmanngasse 5, 1090 Vienna, Austria}\\ 
\vspace{-0.4cm}\\
\texttt{David.Fajman@univie.ac.at}\\

\noindent
\textsc{Klaus Kr\"oncke\\
Department of Mathematics, University of Hamburg,\\
Bundesstrasse 55, 20146 Hamburg, Germany}\\ 
\vspace{-0.4cm}\\
\texttt{klaus.kroencke@uni-hamburg.de}\\

\end{document}